\documentclass[12pt]{article}
\usepackage{amssymb}
\usepackage{color}
\usepackage{indentfirst}
\usepackage{amsmath}
\usepackage{enumerate}
\usepackage{cite}
\usepackage{ntheorem}
\newtheorem{theorem}{Theorem}
\newtheorem{lemma}{Lemma}
\newtheorem{corollary}{Corollary}

\newtheorem{definition}{Definition}

\newtheorem{example}{Example}

\newtheorem*{proof}{Proof}

\begin{document}
\title{The $c$-differential properties of a class of power functions}
\author{\small{ Huan Zhou$^{1}$}, Xiaoni Du\thanks{Corresponding author.}\ $^{1,2,3}$, Wenping Yuan$^{1}$, Xingbin Qiao$^{1}$\\
\small{${}^1$College of Mathematics and Statistics, Northwest Normal University, }\\
\small{ Lanzhou, 730070, China}\\
\small{${}^2$Key Laboratory of Cryptography and Data Analytics,}\\
\small{ Northwest Normal University, 730070, Lanzhou, China }\\
\small{${}^3$Gansu Provincial Research Center for Basic Disciplines of Mathematics and Statistics,}\\
\small{Northwest Normal University, 730070, Lanzhou, China}\\
\small{Email: nwnuzh@126.com}}
\date{}
\maketitle

\begin{abstract}
Power functions with low $c$-differential uniformity have been widely studied not only because of their strong resistance to multiplicative differential attacks, but also low implementation cost in hardware. Furthermore, the $c$-differential spectrum of a function gives a more precise characterization of its $c$-differential properties. Let $f(x)=x^{\frac{p^n+3}{2}}$ be a power function over the finite field $\mathbb{F}_{p^{n}}$, where $p\neq3$ is an odd prime and $n$ is a positive integer. In this paper, for all primes $p\neq3$,  by investigating certain character sums with regard to elliptic curves and computing the number of solutions of a system of equations over $\mathbb{F}_{p^{n}}$, we determine explicitly the $(-1)$-differential spectrum of $f$ with a unified approach. We show that if $p^n \equiv 3 \pmod 4$, then $f$ is a differentially $(-1,3)$-uniform  function except for $p^n\in\{7,19,23\}$ where $f$ is an APcN function, and if $p^n \equiv 1 \pmod 4$, the $(-1)$-differential uniformity of $f$ is equal to $4$. In addition, an upper bound of the $c$-differential uniformity of $f$ is also given.


\end{abstract}

\section{Introduction}\label{A}

Let $\mathbb{F}_{p^n}$ be the finite field with $p^n$ elements, where $p$ is a prime and $n$ is a positive integer. The multiplicative cyclic group of the finite field $\mathbb{F}_{p^n}$ is denoted by $\mathbb{F}_{p^n}^{*}$. Let $\mathbb{F}_{p^n}[x]$ denote the polynomial ring over $\mathbb{F}_{p^{n}}$. Any function $f:\mathbb{F}_{p^n}\rightarrow\mathbb{F}_{p^n}$ can be uniquely represented as a univariate polynomial of degree less than $p^{n}$. Therefore, $f$ can always be seen as a polynomial in $\mathbb{F}_{p^n}[x]$. A polynomial $f\in \mathbb{F}_{p^n}[x]$ is called a permutation polynomial of $\mathbb{F}_{p^n}$ if the induced mapping $x\mapsto f(x)$ is a permutation over $\mathbb{F}_{p^n}$. Permutation polynomials are a very important class of polynomials as they have applications in coding theory and cryptography. Nowadays, designing infinite classes of permutation polynomials over finite fields with good cryptographic properties remains an interesting research topic.

Substitution boxes (S-boxes for short) are important nonlinear building blocks in  block ciphers, which can be seen as  permutation functions over finite fields. To quantify the ability of S-boxes to resist the differential attack\cite{BS1991}, one of the most powerful attacks on block ciphers, Nyberg\cite{N1994} introduced the notion of differential uniformity. More importantly, to estimate the resistance against some variants of differential cryptanalysis, the differential spectrums\cite{BCC2010} of S-boxes are also shown to be of great interest.

In \cite{BCJW2002}, Borisov et al. proposed a new differential on block ciphers, called a multiplicative differential, it is an extension of differential cryptanalysis, which helps to identify the weakness of IDEA ciphers. Motivated by the concept of multiplicative differential, Ellingsen et al. in \cite{EFRST2020} first introduced the (multiplicative) $c$-derivative.
\begin {align*}
_{c}D_{a}f(x)=f(x+a)-cf(x),~\mathrm{for~all}~x \in\mathbb{F}_{p^n}.
\end {align*}
In the same paper they presented a generalized notion of differential uniformity, the so-called $c$-differential uniformity.
\begin{definition}\cite{EFRST2020}
Let $\mathbb{F}_{p^n}$ be the finite field with $p^n$ elements, where $p$ is a prime and $n$ is a positive integer. For a function $f:\mathbb{F}_{p^n}\rightarrow\mathbb{F}_{p^n}$, and $a,b,c\in \mathbb{F}_{p^n}$. Let ${}_{c}\Delta_{f}(a,b) =\#\{x\in \mathbb{F}_{p^n}:f(x+a)-cf(x)=b\}$. The $c$-differential uniformity of $f$ is defined as
$$_{c}\Delta_{f}=\max\{{}_{c}\delta_{f}(a,b):a,b\in \mathbb{F}_{p^n},~ \mathrm{and}~a\neq0~\mathrm{if}~c=1\}.$$
\end{definition}

If $_c\Delta_f=\delta$, then $f$ is called a differentially $(c,\delta)$-uniform function. Especially, $f$ is called a perfect $c$-nonlinear (PcN) function if $_{c}\Delta_{f}=1$, and an almost perfect $c$-nonlinear (APcN) function if $_{c}\Delta_{f}=2$. It is clear that if $c=1$ and $a\neq0$, the $c$-differential uniformity becomes the usual differential uniformity and one use the corresponding notation by omit the symbol $c$. The smaller the value $_{c}\Delta_{f}$ is, the better $f$ resists against multiplicative differential attacks. Thus, the research on cryptographic functions with low $c$-differential uniformity has been a hot issue in recent years, the readers can refer to \cite{EFRST2020,MRSYZ2021,MMM2022,PLZZ2022,RSY2022,
WLZ2021,WZH2022,TLWZTJ2023,Y2022,ZH2021,YZ2022,YMT,WMLZ2023} and the references therein.

For a power function $f(x)=x^d$ with a positive integer $d$, one can easily see that $_{c}\delta_{f}(a,b)={}_{c}\delta_{f}(1,b/a^{d})$, for all $a\in \mathbb{F}_{p^n}^{*}$ and $b\in \mathbb{F}_{p^n}$. For simplicity, denoted ${}_{c}\delta_{f}(1,b)$ by $_{c}\delta_{f}(b)$ with $b\in \mathbb{F}_{p^n}$. More precisely, it was proved in \cite[Lemma 1]{MRSYZ2021} that the $c$-differential uniformity of $f(x)=x^d$ is given by
\begin{equation}\label{cf}
_c\Delta_f=\max\big\{ \{_c\delta_f(b) : b \in \mathbb{F}_{p^n} \} \cup \{\gcd(d,p^n-1)\} \big\}.
\end{equation}

Moreover, as a generalization of the differential spectrum, the $c$-differential spectrum of a power function is defined as follows.
\begin{definition}\cite{WZH2022}
Let $f(x)=x^{d}$ with $c$-differential uniformity $_{c}\Delta_{f}$. Denote $_{c}\omega_{i}=\#\{b\in\mathbb{F}_{p^n}: {}_{c}\delta_{f}(b)=i\}$, for each $0\leq i\leq {}_{c}\Delta_{f}$. Then the $c$-differential spectrum of $f$ is defined as the multi-set
$$
\mathbb{S}=\{_{c}\omega_{i}>0:0\leq i\leq{}_{c}\Delta_{f}\}.
$$
\end{definition}

Generally speaking, it is difficult to determine the $c$-differential spectrum of power functions. To the best of our knowledge, only a few classes of power functions over odd characteristic finite fields have a nontrivial $c$-differential spectrum. The known results on power functions $f$ over $\mathbb{F}_{p^n}$, for which $c$-differential spectrum has been determined are summarized in Table \ref{biao1}, where Tr$_{n}(\cdot)$ is the absolute trace mapping from $\mathbb{F}_{2^n}$ to $\mathbb{F}_{2}$.
\begin{table}
  \scriptsize
  \renewcommand{\arraystretch}{1.5}
  \caption{Power functions $f(x)=x^{d}\in \mathbb{F}_{p^n}[x]$ with known $c$-differential spectrum}\label{biao1}
  \centering
\begin{tabular}{|c|c|c|c|}
  \hline
  $d$ & Conditions & $_{c}\Delta_{f}$ & Ref. \\
  \hline
  $2^{3m}+2^{2m}+2^{m}-1$ & $p=2,~n=4m,~0,1\neq c\in \mathbb{F}_{2^n},~c^{1+2^{2m}}=1$~ & 2 & \cite{TLWZTJ2023} \\
  \hline
  $p^{k}+1$&~$p$~is odd,~$\gcd(n,k)=e,~1\neq c\in \mathbb{F}_{p^{e}}$~and~$\frac{n}{e}$~is odd& 2&\cite{WZH2022}\\
  \hline
  $p^{k}+1$&~$p$~is odd,~$\gcd(n,k)=e,~c\notin \mathbb{F}_{p^{e}}$,~$n$~is even,~$k=\frac{n}{2}$~&2& \cite{WZH2022} \\
  \hline
  $2^{n}-2$&~$p=2,c\neq0$,~Tr$_{n}(c)=$Tr$_{n}(c^{-1})=1$&2& \cite{YZ2022} \\
  \hline
  $2^{n}-2$&~$p=2,c\neq0$,~Tr$_{n}(c)=0$~or~Tr$_{n}(c^{-1})=0$&3& \cite{YZ2022} \\
  \hline
  $p^{n}-2$&~$p$~is odd,~$c\neq0,1,4,4^{-1}$~& 2~or~3 & \cite{YZ2022} \\
  \hline
  $3^{n}-3$&~$p=3,~c=-1$,~any~$n$~ & 4~or~6 & \cite{YZ2022} \\
  \hline
  $\frac{p^k+1}{2}$&~$p$~is odd, ~$c=-1,~\gcd(n,k)=1$,~$\frac{2n}{\gcd(2n,k)}$~is even& $\frac{p+1}{2}$ & \cite{YZ2022} \\
  \hline
  $\frac{p^n+1}{2}$&~any $p$,~$c\neq-1$~&2,4,~$\frac{p^n+3}{4}$~or~$\frac{p^n+5}{4}$& \cite{RSY2022} \\
  \hline
  $\frac{5^n-3}{2}$&$p=5,~c=-1,~n\geq2$&2& \cite{YZ2022}~and~\cite{PLZZ2022} \\
  \hline
  $\frac{p^n-3}{2}$&$c=-1,~p^n~\equiv3~(\bmod~4),~p^{n}>3$&2~or ~4&\cite{YMT}\\
  \hline
  $\frac{3^n+3}{2}$&~$p=3$,~$c=-1$,~$n\geq2$~is even~& 2 & \cite{YZ2022} \\
  \hline
  $\frac{p^n+3}{2}$&~$c=-1$,~$p^n~\equiv1~(\bmod~4)$,~$p\neq3$~& 4 & This~paper \\
  \hline
  $\frac{p^n+3}{2}$&~$c=-1$,~$p^n~\equiv3~(\bmod~4)$,~$p\neq3$~& 3 & This~paper \\
  \hline
\end{tabular}
\end{table}

Throughout this paper, let $f(x)=x^{\frac{p^{n}+3}{2}}\in\mathbb{F}_{p^n}[x]$, where $p$ is an odd prime and $n$ is a positive integer. We should mention that when $p=3$, $f$ is P$c$N if $n\geq3$ is odd, and is AP$c$N if $n\geq2$ is even \cite{MRSYZ2021} with its $(-1)$-differential spectrum being given in \cite[Theorem 9]{YZ2022}. When $p>3$, the $(-1)$-differential uniformity of $f$ was discussed by Mesnager et al. in \cite[Theorem 11]{MRSYZ2021} and they proved that $_{-1}\Delta_{f}\leq 4$ if $p^n \equiv1 \pmod 4$ and $_{-1}\Delta_{f}\leq3$ if $p^n\equiv3 \pmod4$.

Inspired by the works above, in this paper, we mainly study the $(-1)$-differential spectrum of $f(x)=x^{\frac{p^{n}+3}{2}}$ for all positive integer $n$ when $p\neq3$. More precisely, we prove that $_{-1}\Delta_{f}=3$ if $p^n\equiv3 \pmod 4$ except for $p\in\{7,19,23\}$ where $f$ is an APcN function, and $_{-1}\Delta_{f}=4$ if $p^n\equiv 1 \pmod 4$, that is, the upper bound of the $(-1)$-differential uniformity of $f$ given by Mesnager et al. can be achieved.

The rest of this paper is organized as follows. In
Section \ref{B}, we introduce some basic notation about quadratic character and results on quadratic character sums, which will be employed in the sequel. In Section \ref{C}, we first determine the $(-1)$-differential spectrum of $f$, and some examples are also presented. In Section \ref{D}, we give an upper bound of the $c$-differential uniformity of $f$. Section \ref{E} concludes this paper.
\section{Preliminaries}\label{B}
In this section, we mainly introduce some basic results on quadratic character sums over $\mathbb{F}_{p^n}$. Let $\eta$ be the quadratic character over $\mathbb{F}_{p^n}$, i.e., for any $x\in \mathbb{F}_{p^{n}}$,
\begin{align*}
\eta(x)=x^{\frac{p^{n}-1}{2}}=
\left\{
    \begin{array}{ll}
    1, & \hbox{if $x$ is a square,}\\
    0, & \hbox{if $x=0$,}\\
    -1, & \hbox{if $x$ is a nonsquare.}
\end{array}
\right.
\end{align*}
It is well-known that $\sum\limits_{x\in \mathbb{F}_{p^n}}\eta(x)=0$, and $\eta(-1)=1$ (resp. $-1$) if $p^n\equiv1~(\bmod~4)$  (resp. $p^n\equiv3~(\bmod~4)$), which are used extensively in the calculations of character sums.

We consider now sums involving the quadratic character of the form
\begin{equation*}
\sum_{x\in \mathbb{F}_{{p^{n}}}}\eta(f(x))
\end{equation*}
with $f(x)\in \mathbb{F}_{p^{n}}[x]$. Recall that for $\deg(f(x))=1$,  the  sums above is trivial, and for $\deg(f(x))=2$, the following explicit formula was established in \cite{LN1997} and \cite{MMM2022}, respectively.
\begin{lemma}\cite{LN1997}\label{1}
Let $f(x)=a_{2}x^{2}+a_{1}x+a_{0}\in \mathbb{F}_{p^{n}}[x]$, $p$ is odd, and $a_{2}\neq 0$. Put $\Delta:=a_{1}^{2}-4a_{0}a_{2}$ be the discriminant of $f(x)=0$ and let $\eta$ be the quadratic character of $\mathbb{F}_{p^{n}}$. Then
\begin{align*}
\sum_{x\in \mathbb{F}_{{p^{n}}}}\eta(f(x))=
  \left\{
    \begin{array}{ll}
    -\eta(a_{2}), & \hbox{\rm{if} $\Delta\neq0$,}\\
    (p^{n}-1)\eta(a_{2}),  & \hbox{\rm{if} $\Delta=0$.}
    \end{array}
  \right.
\end{align*}
\end{lemma}

\begin{lemma}\cite{MMM2022}\label{2}
Let $f(x)=a_2x^{2}+a_1x+a_0\in \mathbb{F}_{p^{n}}[x]$ with $p$ odd and $a_2\neq0$. Then the equation $f(x)=0$ has two (resp. one) solutions in $\mathbb{F}_{p^{n}}$ if and only if the discriminant $\Delta=a_{1}^{2}-4a_{0}a_{2}$ is a nonzero (resp. zero) square in $\mathbb{F}_{p^{n}}$. That is to say, the number of solutions of $f$ is $1+\eta(\Delta)$.
\end{lemma}

For $\deg(f(x))\geq 3$, it is challenging to obtain an explicit formula for the character sum $\sum\limits_{x\in \mathbb{F}_{p^{n}}}\eta(f(x))$. However, when $\deg(f(x))=3$, such a sum can be computed by considering $\mathbb{F}_{p^n}$-rational points of elliptic curves over $\mathbb{F}_{p}$. More specifically, we denote $\lambda_{p,n}$ as
$$\lambda_{p,n}=\sum_{x\in \mathbb{F}_{p^{n}}}\eta(f(x)).$$

To calculate $\lambda_{p,n}$, we will use some elementary concepts from the theory of elliptic curves \cite{S2009}. Let $E/\mathbb{F}_{p}$ be the elliptic curve over $\mathbb{F}_{p}$
$$E:y^{2}=f(x)$$
and $N_{p,n}$ denote the number of $\mathbb{F}_{p^{n}}$-rational points (remember the extra point at infinity) on the curve $E/\mathbb{F}_{p}$. From the results in \cite[P.139, P.142]{S2009}, we have, for all $n\geq1$,
$$N_{p,n}=p^n+1+\lambda_{p,n}$$
with
\begin{align*}
\lambda_{p,n}=-\alpha^{n}-\beta^{n},
\end{align*}
where $\alpha$ and $\beta$ are the two conjugate complex roots of the polynomial $T^{2}+\lambda_{p,1}T+p$.
We obtain an explicit and efficient formula of $\lambda_{p,n}$.

Define the following two specific character sums
\begin{align}\label{sum1}
\lambda^{(1)}_{p,n}=\sum_{x\in \mathbb{F}_{p^{n}}}\eta(x(x+1)(x-3))
\end{align}
and
\begin{align}\label{sum2}
\lambda^{(2)}_{p,n}=\sum_{x\in \mathbb{F}_{p^{n}}}\eta(x(x+1)(x-2)),
\end{align}
which play a significant role in our main results. As we will see later, the determination of the $(-1)$-differential spectrum of $f(x)=x^{\frac{p^{n}+3}{2}}$ over $\mathbb{F}_{p^{n}}$ heavily depend on the computations of $\lambda^{(1)}_{p,n}$ and $\lambda^{(2)}_{p,n}$.

In the following examples, we will give the exact values of $\lambda^{(1)}_{p,n}$ and $\lambda^{(2)}_{p,n}$, respectively, for specific values of $p$.

\begin{example}
Let $p=5$. We can obtain $\lambda^{(1)}_{5,1}=2$ by Magma program. So we have $\alpha,\beta=-1\pm2\sqrt{-1}$. Hence $\lambda^{(1)}_{5,n}=-(-1+2\sqrt{-1})^{n}-(-1-2\sqrt{-1})^{n}$.
\end{example}

\begin{example}
Let $p=13$. We can obtain $\lambda^{(2)}_{13,1}=2$ by Magma program. So we have $\alpha,\beta=-1\pm2\sqrt{-3}$. Hence $\lambda^{(2)}_{13,n}=-(-1+2\sqrt{-3})^{n}-(-1-2\sqrt{-3})^{n}$.
\end{example}

In addition, we have the following bound on $\lambda^{(i)}_{p,n}$ for $i=1,2$.
\begin{lemma}\cite{S2009}\label{bound}
With the notation as above, we have $|\lambda^{(i)}_{p,n}|\leq 2p^{\frac{n}{2}}$, for $i=1,2$.
\end{lemma}

In the end of this section, we present the following results concerning the exact values specific character sums used in Section \ref{C}.
\begin{lemma}\label{4}
If $p^n\equiv3\pmod4$ and $p\neq 3$, then relevant consequences as follows,

$1)\sum\limits_{x\in
\mathbb{F}_{p^{n}}}\eta\left((2x-2)(2x+1)(2x+2)\right)=\lambda_{p, n}^{(1)}$,

$2)\sum\limits_{x\in
\mathbb{F}_{p^{n}}} \eta\left((2x-2)(2x-1)(2x+2)\right)=-\lambda_{p, n}^{(1)}$,

$3)\sum\limits_{x\in
\mathbb{F}_{p^{n}}}\eta((2x-2)(2x-1)(2x+1)(2x+2))=\lambda_{p, n}^{(1)}-1$,

$4)\sum\limits_{x\in
\mathbb{F}_{p^{n}}} \eta\left((2x-1)(2x+1)(2x+2)\right)=\lambda_{p, n}^{(2)}$,

$5)\sum\limits_{x\in
\mathbb{F}_{p^{n}}} \eta\left((2x-2)(2x-1)(2x+1)\right)=-\lambda_{p, n}^{(2)}$.
\end{lemma}

\begin{proof}
1) It is clear that $2x+1$ is a permutation of $\mathbb{F}_{p^{n}}$. Set $t=2x+1$. The left side of equation 1) leads to
$$
\sum\limits_{t\in\mathbb{F}_{p^{n}}} \eta(t(t+1)(t-3))=\lambda_{p, n}^{(1)}.
$$

2) Set $t=2x-1$. The left side of equation 2) leads to
\begin{align*}
&\sum\limits_{t\in\mathbb{F}_{p^{n}}}\eta(t(t-1)(t+3))\\
=&\sum\limits_{t\in\mathbb{F}_{p^{n}}}\eta((-t)(-t-1)(-t+3))\\
=&-\sum\limits_{t\in\mathbb{F}_{p^{n}}}\eta(t(t+1)(t-3))\\
=&-\lambda_{p,n}^{(1)}
\end{align*}
since $\eta(-1)=-1$.

3) First observe that the left side of equation 3) can be written as
\begin{align}\label{ux}
\sum\limits_{4x^{2}-1 \neq 0} \eta\Big(\frac{4x^{2}-4}{4x^{2}-1}\Big).
\end{align}

If we set $\frac{4x^{2}-4}{4x^{2}-1}=u$, then $u$ and $x$ satisfy
$$
(4-4u)x^{2}+u-4=0.
$$

When $u\neq 1$, it is a quadratic equation with connection on the variable $x$, with discriminant $\Delta=16(u-4)(u-1)$. Therefore, according to Lemma \ref{2} for all $\ u\neq1$, it corresponds to $(1+\eta(\Delta))$ $x$'s, and then
\begin{align*}
\mathrm{Eq.}\eqref{ux}
=&\sum\limits_{u \in \mathbb{F}_{p^{n}}} \eta(u)(1+\eta(\Delta))-\eta(1)\\
=&\sum\limits_{u \in \mathbb{F}_{p^{n}}} \eta(u(u-4)(u-1))-1\\ =&\sum\limits_{t \in \mathbb{F}_{p^{n}}} \eta(t(t+1)(t-3))-1\\
=&\lambda_{p, n}^{(1)}-1.
\end{align*}

5) Set $t=2x-1$. The left side of equation 5) leads to
\begin{align*}
&\sum\limits_{t \in \mathbb{F}_{p^{n}}} \eta(t(t-1)(t+2))\\
=&\sum\limits_{t \in \mathbb{F}_{p^{n}}} \eta((-t)(-t-1)(-t+2))\\
=&-\sum\limits_{t \in \mathbb{F}_{p^{n}}} \eta(t(t+1)(t-2))\\
=&-\lambda_{p, n}^{(2)}.
\end{align*}
\end{proof}
The proof of following lemma is very similar to that of \cite[Lemma 5]{YMT2023}, we will no longer prove that.
\begin{lemma}\label{sums}
Let $\lambda^{(1)}_{p,n}$ and $\lambda^{(2)}_{p,n}$ be defined as above. If $p^n\equiv1\pmod4$  and $p\neq 3$, then
 we have

$(1)\sum\limits_{x\in \mathbb{F}_{p^{n}}}\eta(x(x^2+x+1))=\lambda^{(1)}_{p,n}$,

$(2)\sum\limits_{x\in \mathbb{F}_{p^{n}}}\eta((x+1)(x^2+x+1))=\lambda^{(1)}_{p,n}$,

$(3)\sum\limits_{x\in \mathbb{F}_{{p^{n}}}}\eta((x^2+x)(x^2+x+1))=\lambda^{(1)}_{p,n}-1$,

$(4)\sum\limits_{x\in \mathbb{F}_{p^{n}}}\eta(x(3x^2+2x+3))=\lambda^{(2)}_{p,n}$,

$(5)\sum\limits_{x\in \mathbb{F}_{{p^{n}}}}\eta((x^2+x+1)(3x^2+2x+3))=\lambda^{(2)}_{p,n}-\eta(3)$,

$(6)\sum\limits_{x\in \mathbb{F}_{p^{n}}}\eta(x(x^2+x+1)(3x^2+2x+3))=2\lambda^{(1)}_{p,n}$.\\
\end{lemma}

\section{The $c$-differential properties of $f(x)=x^{\frac{p^n+3}{2}}$ over $\mathbb{F}_{p^n}$}\label{C}
In this section, we are about to determine the $(-1)$-differential spectrum of $f$ explicitly.

We first introduce some properties of the $c$-differential spectrum of power function presented by
Yan and Zhang in \cite{YZ2022},  which will be used to examine the $(-1)$-differential spectrum of $f$.
\begin{lemma}\cite{YZ2022}
Let $f(x)=x^{d}$ be a power function over $\mathbb{F}_{p^n}$ with $c$-differential uniformity $_{c}\Delta_{f}$ for some $1\neq c\in\mathbb{F}_{p^n}$, where $d$ is a positive integer. Then we have
\begin{align}\label{c-omegaiomega}
\sum_{i=0}^{_{c}\Delta_{f}}{_{c}\omega_{i}}
=\sum_{i=0}^{_{c}{\Delta_{f}}} (i \cdot {_{c}{\omega_{i}}})=p^{n}.
\end{align}
Moreover,
\begin{equation}\label{c-i^2omega}
\sum\limits_{i = 0}^{{}_c\Delta_f}(i^2\cdot {}_c\omega _i)=\frac{_cN_4-1}{p^n-1}-\gcd(d,p^{n}-1),
\end{equation}
where
\begin{equation}\label{n4}
{}_cN_{4}=\#\left\{ {(x_1,x_2,x_3,x_4) \in (\mathbb{F}_{p^n})^4 : \Bigg\{ \begin{array}{ll}
x_1 - x_2 + x_3 - x_4 &= 0\\
x_1^d - cx_2^d + cx_3^d - x_4^d &= 0
\end{array} }  \right\}.
\end{equation}
\end{lemma}
\subsection{$p^n \equiv 3 \pmod 4$ and $p\neq3$ }
In this subsection, we give the $(-1)$-differential spectrum of the power function $f(x)=x^{\frac{p^n+3}{2}}$ over $\mathbb{F}_{p^n}$, $p^n \equiv 3 \pmod 4$ and $p\neq3$.
Before that, we give some notions.

For any square $\alpha\in\mathbb{F}_{p^n}^{*}$,~$\sqrt\alpha$ denotes either solution of the equation $x^{2}=\alpha$ in $\mathbb{F}_{p^n}$. Let $C_{0}$ and $C_{1}$ denote the sets of squares and nonsquares in $\mathbb{F}_{p^n}^{*}$, respectively. The cyclotomic number $(i,j)$ is defined as the cardinality of
the set $C_{i,j}=\{{x\in\mathbb{F}_{p^n}\backslash\{{0,-1}\}:x\in C_{i},x+1\in C_{j}}\}$, where $i,j\in\{{0,1}\}$.

\begin{theorem}\label{Thm3}
Let $f(x)=x^{\frac{p^n+3}{2}}$ over $\mathbb{F}_{p^n}$, where $p^n \equiv 3 \pmod 4$ and $p\neq3$. Then the $(-1)$-differential spectrum of $f$ is
\begin{align*}
\mathbb S=\{&_{-1}\omega_{0}=\frac{1}{4}(p^{n}+\lambda_{p, n}^{(1)}+1),~{}_{-1}\omega_{1}=\frac{1}{16}(9 p^{n}-5\lambda_{p, n}^{(1)}+2\lambda_{p, n}^{(2)}-7),\\
&{}_{-1}\omega_{2}=\frac{1}{8}(p^{n}-\lambda_{p,n}^{(1)}-2\lambda_{p, n}^{(2)}+1),
~{}_{-1}\omega_{3}=\frac{1}{16}(p^{n}+3\lambda_{p, n}^{(1)}+2\lambda_{p, n}^{(2)}+1)\}.
\end{align*}
\begin{proof}

In order to determine the $(-1)$-differential spectrum of $f$, the number of solutions $N_{f}(b)$ of the equation
\begin{equation}\label{7}
(x+1)^{d}+x^{d}=b,~\mathrm{for~all}~b\in\mathbb{F}_{p^{n}}
\end{equation}
needs to be determined.

Evidently, if $x=0$, then $b=1$, and if $x=-1$, then $b=-1$ since $d=\frac{p^{n}+3}{2}$ is odd. Now we always assume $x\in \mathbb{F}_{p^n}^{*}\backslash\{-1\}$, Eq.\eqref{7} can be written as
\begin{equation}\label{8}
(\eta(x+1)+\eta(x)) x^{2}+2 \eta(x+1) x+\eta(x+1)-b=0.
\end{equation}
Note that $\mathbb{F}_{p^{n}}^{*}\backslash\{-1\}=\mathcal{C}_{0,0} \cup \mathcal{C}_{0,1} \cup \mathcal{C}_{1,0} \cup \mathcal{C}_{1,1}$. The following four cases are discussed.

(1) $x \in \mathcal{C}_{0,0}$, that is $\eta(x)=\eta(x+1)=1$. Eq.\eqref{8} becomes
\begin{equation}\label{9}
2x^{2}+2x+1-b=0
\end{equation}
with discriminant $\Delta_{1}=8b-4$.

If $\Delta_{1}=0$, then $b=\frac{1}{2}$, by Eq.\eqref{9}, we get $x=-\frac{1}{2}$. So $\eta(x+1)=\eta(-\frac{1}{2}+1)
=\eta(\frac{1}{2})=\eta(-x)=-\eta(x)$, i.e., $\eta(x)\neq\eta(x+1)$, which means that $x=-\frac{1}{2} \notin \mathcal{C}_{0,0}$.

If $\eta\left(\Delta_{1}\right)=1$, we can acquire two solutions of Eq.\eqref{9} over $\mathbb{F}_{p^n}^{*} \backslash\{-1\}$, denote by $x_{1}=\frac{-1 + \sqrt{2 b-1}}{2}$, $x_{1}^{\prime}=\frac{-1 - \sqrt{2 b-1}}{2}$. Apparently, $x_{1}+1=-x_{1}^{\prime}$, $x_{1}^{\prime}+1=-x_{1}$ and  $x_{1}\left(x_{1}+1\right)=\frac{b-1}{2}$. Since $\eta(-1)=-1$, we assert that Eq.\eqref{8} has at most one solution in $\mathcal{C}_{0,0}$ when $b\in D_{1}$, where
$$D_{1}=\{b \in\mathbb{F}_{p^{n}}: \eta(2b-1)=\eta(2b-2)=1\}.$$

(2) $x \in \mathcal{C}_{0,1}$, that is $\eta(x)=1,~\eta(x+1)=-1$. Eq.\eqref{8} becomes $2x+1+b=0$, one can be deduced that $x=\frac{-1-b}{2}$, $x+1=\frac{1-b}{2}$. Hence Eq.\eqref{8} has at most one solution when $b\in D_{2}$, where
$$D_{2} =\{b \in \mathbb{F}_{p^{n}}: \eta(2b-2)=1,~\eta(2b+2) =-1\}.$$

(3) $x \in \mathcal{C}_{1,0}$, that is $\eta(x)=-1, \eta(x+1)=1$. Eq.\eqref{8} becomes $2x+1-b=0$, then $x=\frac{b-1}{2}$, $x+1=\frac{b+1}{2}$. Therefore, Eq.\eqref{8} has at most one solution in $\mathcal{C}_{1,0}$ if $b\in D_{3}$, where
$$D_{3}=\{b \in \mathbb{F}_{p^n}:\eta(2b-2)=-1,~\eta(2b+2) =1\}.$$

(4) $x \in \mathcal{C}_{1,1}$, that is $\eta(x)=\eta(x+1)=-1$. Eq.\eqref{8} becomes
\begin{equation}\label{10}
2x^{2}+2x+1+b=0,
\end{equation}
with discriminant $\Delta_{2}=-4-8b$.

If $\Delta_{2}=0$, then $b=-\frac{1}{2}$ and $x=-\frac{1}{2}$. So $\eta(x+1)=\eta(\frac{1}{2})=\eta(-x)=-\eta(x)$, i.e., $\eta(x)\neq\eta(x+1)$, which means that $x=-\frac{1}{2} \notin \mathcal{C}_{1,1}$.

If $\eta\left(\Delta_{2}\right)=1$, we can acquire two solutions of Eq.\eqref{10} over $\mathbb{F}_{p^n}^{*} \backslash\{-1\}$, denote by $x_{2}=\frac{-1 + \sqrt{-2 b-1}}{2}$, $x_{2}^{\prime}=\frac{-1 - \sqrt{-2 b-1}}{2}$. Obviously, $x_{2}+1=-x_{2}^{\prime}$, $x_{2}^{\prime}+1=-x_{2}$ and $x_{2}\left(x_{2}+1\right)=-\frac{1+b}{2}$. Since $\eta(-1)=-1$, we assert that Eq.\eqref{8} has at most one solution in $\mathcal{C}_{1,1}$ when $b\in D_{4}$, where
$$D_{4} =\{b \in \mathbb{F}_{p^n} : \eta(2b+1)=\eta(2b+2)=-1\}.$$

We can get that $b=\pm 1 \notin \bigcup\limits_{i=1}^{4} D_{i}$, $D_{3} \cap\left(D_{1} \cup D_{2} \cup D_{4}\right)=\emptyset$ and $D_{1} \cap D_{4}=D_{1} \cap D_{2} \cap D_{4}$, The discussion above can be summed up in Table \ref{biao2}.
\begin{table}[h]
  \small
  \centering
  \caption{Solutions to Eq.\eqref{7} in $\mathbb{F}_{p^{n}}$
  \label{biao2}}
\begin{tabular}{|c|c|}
  \hline
  $b$ & solutions \\
  \hline
  $b=1$ & $x=0$ \\
  \hline
  $b=-1$ & $x=-1$ \\
  \hline
  $b\in D_{1}$ & $x\in A\subset\{\frac{1+\sqrt{2b-1}}{2}, \frac{1-\sqrt{2b-1}}{2}\}$ \\
  \hline
  $b\in D_{2}$ & $x=-\frac{b+1}{2}$ \\
  \hline
  $b\in D_{3}$ & $x=\frac{b-1}{2}$ \\
  \hline
  $b\in D_{4}$ & $x\in B\subset\{\frac{-1+\sqrt{-2b-1}}{2}, \frac{-1-\sqrt{-2b-1}}{2}\}$ \\
  \hline
\end{tabular}
\end{table}

Hence Eq.\eqref{7} has at most three solutions in $\mathbb{F}_{p^{n}}$, and the number of the solutions of Eq.\eqref{7} is
\begin{align*}
N_{f}(b)=\left\{\begin{array}{ll}
3, & \mathrm{if}~b \in D_{1} \cap D_{2} \cap D_{4}, \\
2, & \mathrm{if}~b \in(\left(D_{1} \cap D_{2}\right) \backslash D_{4})\cup (\left(D_{2} \cap D_{4}\right) \backslash D_{1}), \\
1, &  \mathrm{if}~b=\pm 1,~\mathrm{or}~b \in D_{3},\mathrm{or}~b\in (D_{2} \backslash\left(D_{1} \cup D_{4}\right))\\
   & ~\cup (D_{1}\backslash\left(D_{2} \cup D_{4}\right))\cup(D_{4} \backslash\left(D_{1} \cup D_{2}\right)),\\
0, & \mathrm{otherwise}.
\end{array}\right.
\end{align*}

Now we calculate the cardinality of the sets $D_{1}\cap D_{2}\cap D_{4}$, $(D_{1} \cap D_{2}) \backslash D_{4}$, $(D_{2} \cap D_{4})\backslash D_{1}$ to determine $(-1)$-differential spectrum of $f$.

First, note that
$$D_{1}\cap D_{2}\cap D_{4}=\{b \in \mathbb{F}_{p^n}: \eta(2b-2)=\eta(2b-1)=1,~\eta(2b+1)=\eta(2b+2)=-1\}.$$
By the definition of the $(-1)$-differential spectrum of $f$, we have
\begin{align*}
_{-1}\omega_{3}=&\#(D_{1}\cap D_{2}\cap D_{4})\\
=&\frac{1}{16}\sum_{b\in\mathbb{F}_{p^{n}}}((1+\eta(2b-2))(1+\eta(2b-1))
(1-\eta(2b+1))(1-\eta(2 b+2)))\\
&-\frac{1}{16}\sum_{b\in\{\pm 1,\pm\frac{1}{2}\}}
((1+\eta(2b-2))(1+\eta(2b-1))(1-\eta(2b+1))(1-\eta(2 b+2)))\\
=&M_{1}-M_{2}.
\end{align*}
where $M_{1}$ and $M_{2}$ denote by the two summations above respectively.

It can be easy to see that $M_{2}=0$. Expanding the expression of $M_1$,  then we have that
$$M_1=\frac{1}{16}(p^{n}+3\lambda_{p, n}^{(1)}+2\lambda_{p, n}^{(2)}+1)$$
after a calculation based on the fact in Lemmas \ref{1}, \ref{2} and \ref{4}.\\
Therefore, $_{-1}\omega_{3}=M_{1}-M_{2}
=\frac{1}{16}(p^{n}+3\lambda_{p,n}^{(1)}+2\lambda_{p,n}^{(2)}+1)$.

Similarly, we have
$$(D_{1}\cap D_{2})\backslash D_{4}=\{b\in\mathbb{F}_{p^{n}}:\eta(2b-2)=\eta(2b\pm1)=1,~\eta(2b+2)=-1\},$$
$$(D_{2}\cap D_{4})\backslash D_{1}=\{b \in \mathbb{F}_{p^{n}}: \eta(2b+2)=\eta(2b\pm1)=-1,~\eta(2b-2)=1\},$$
and
\begin{align*}
_{-1}\omega_{2}=&\#((D_{1} \cap D_{2}) \backslash D_{4})+
\#((D_{2}\cap D_{4})\backslash D_{1})\\
=&\frac{1}{16}(p^{n}-\lambda_{p, n}^{(1)}-2\lambda_{p, n}^{(2)}+1)+\frac{1}{16}(p^{n}-\lambda_{p, n}^{(1)}-2\lambda_{p, n}^{(2)}+1)\\
=&\frac{1}{8}(p^{n}-\lambda_{p, n}^{(1)}-2\lambda_{p, n}^{(2)}+1).
\end{align*}
Plugging $_{-1}\omega_{2}$ and $_{-1}\omega_{3}$ into Eq.\eqref{c-omegaiomega}, we get
\begin{align*}
\left\{\begin{array}{ll}
_{-1}\omega_{1}=\frac{1}{16}(9p^{n}-5\lambda_{p, n}^{(1)}+2 \lambda_{p, n}^{(2)}-7),\\
_{-1}\omega_{0}=\frac{1}{4}(p^{n}+\lambda_{p, n}^{(1)}+1).
\end{array}\right.
\end{align*}

In conclusion, which completes the proof.
\end{proof}
\end{theorem}

\begin{corollary}
Let $f(x)=x^{\frac{p^n+3}{2}}$ over $\mathbb{F}_{p^n}$, where $p^n \equiv 3 \pmod 4$ and $p\neq3$ is an odd prime. Then
\begin{align*}
_{-1}\Delta_{f}&=
\left\{\begin{array}{ll}
2, & \hbox{if $p^n\in\{7,19,23\}$,}\\
3, & \hbox{$\mathrm{otherwise}$.}
\end{array}\right.
\end{align*}
\end{corollary}

\begin{proof}
By Lemma \ref{bound}, we obtain
\begin{align*}
_{-1}\omega_{3}=\frac{1}{16}(p^{n}+3\lambda_{p,n}^{(1)}
+2\lambda_{p,n}^{(2)}+1)
\geq\frac{1}{16}(p^{n}-10p^{\frac{n}{2}}+1)>0
\end{align*}
when $p^n\geq103$. Moreover, after calculation by Magma program, we find that the $_{-1}\Delta_{f}=2$ (resp. 3) for $p^n\in\{7,19,23\}$ (resp. $p^n\in\{11,31,43,47,59,67,71,83\}$).

Therefore, we get $f$ is an APcN power function when $p^n\in\{7,19,23\}$, and a differentially $(-1,3)$-uniform power function otherwise.
\end{proof}

There are some examples as follows, which are consistent with that computed directly by Magma program.
\begin{example}
Let $n=1$. If $p=7,19,23$, then $f(x)=x^\frac{p^n+3}{2}$ over $\mathbb{F}_{p}$ is APcN for $c=-1$ with $(-1)$-differential spectrum
\begin{align*}
\mathbb{S}&=\{_{-1}\omega_{0}=2,~{}_{-1}\omega_{1}=3,
~{}_{-1}\omega_{2}=2\},\\
\mathbb{S}&=\{_{-1}\omega_{0}=4,~{}_{-1}\omega_{1}=11,
~{}_{-1}\omega_{2}=4\},\\
\mathbb{S}&=\{_{-1}\omega_{0}=4,~{}_{-1}\omega_{1}=15,
~{}_{-1}\omega_{2}=4\},
\end{align*}
respectively.
\end{example}
\begin{example}
Let $p=7$ and $n=3$. Then $f(x)=x^{173}$ over $\mathbb{F}_{7^{3}}$ is differentially $(-1,3)$-uniform with $(-1)$-differential spectrum
$$\mathbb{S}=\{_{-1}\omega_{0}=86,~{}_{-1}\omega_{1}=195,
~{}_{-1}\omega_{2}=38,~{}_{-1}\omega_{3}=24\}.$$
\end{example}

\subsection{$p^n \equiv 1 \pmod 4$ and $p\neq3$ }
In this subsection, we will focus on studying the $(-1)$-differential spectrum of the power function $f(x)=x^{d}$ over $\mathbb{F}_{p^n}$, where $d=\frac{p^n+3}{2}$, $p^n \equiv 1 \pmod 4$ and $p\neq3$. It is clear that
\begin{equation*}
\gcd(\frac{p^{n}+3}{2},~p^n-1)=2\ \mathrm{or}\ 4.
\end{equation*}

This subsection adopts the method from \cite[P.11-13]{YMT2023}, and we only provide a sketch due to the process is quite complicated and miscellaneous.

Note that $x_{0}$ is a solution of $(x+1)^d+x^d=b$ if and only if $-x_{0}-1$ is a solution of $(x+1)^d+x^d=b$ since $d$ is even. We assert that $_{-1}\delta_{f}(b)$ is even except for
\begin{equation*}
b=(-\frac{1}{2}+1)^d+(-\frac{1}{2})^d=\frac{\eta(2)}{2}.
\end{equation*}

In the following, we will investigate the value of $_{-1}\delta_{f}(\frac{\eta(2)}{2})$. Since the proof of Lemma \ref{deta2} is similar to that of \cite[Lemma 8]{YMT2023}, we omit the proof here.

\begin{lemma}\label{deta2}
With the notation as above, we have	
\begin{align*}
_{-1}\delta_{f}(\frac{\eta(2)}{2})&=\left\{\begin{array}{ll}
3, & \hbox{if  $\eta(2)=\eta(3)=1$, $\eta(\frac{-1+\sqrt{-2}}{2})=-1$ or $\eta(2)=\eta(3)=-1$,}\\
1, & \hbox{$\mathrm{otherwise}$.}
\end{array}\right.
\end{align*}
\end{lemma}

From the discussion above, one can immediately derive the values of ${}_{-1}\omega_1$ and ${}_{-1}\omega_3$ in the following corollary.

\begin{corollary}\label{-1-omega}
With the notation as above, we have ${}_{-1}\omega_1=0, {}_{-1}\omega_3=1$ if $\eta(2)=\eta(3)=1$, $\eta(\frac{-1+\sqrt{-2}}{2})=-1$ or $\eta(2)=\eta(3)=-1$, and ${}_{-1}\omega_1=1, {}_{-1}\omega_3=0$ otherwise.
\end{corollary}

To determine the $(-1)$-differential spectrum of $f$, it remains to calculate $_{-1}N_{4}$ in Eq.\eqref{n4}, and this calculation method is similar to that of \cite[Theorem 11]{YMT2023}.

Moreover, when $c=-1$, Eq.\eqref{n4} can be rewritten as
\begin{equation*}
\left\{\begin{array}{ll}
x_1 - (-x_3) + (-x_2) - x_4 &= 0, \\
x_1^d - (-x_3)^d + (-x_2)^d - x_4^d &= 0,
\end{array}\right.
\end{equation*}
so we use Lemma \ref{sums} to calculate that ${}_{-1}N_{4}$ is \[{}_{-1}N_{4}=1+\frac{1}{8}(p^n-1)(21p^{n}+7\lambda^{(1)}_{p,n}-2\lambda^{(2)}_{p,n}+13),\]
where $\lambda^{(1)}_{p,n}$ and $\lambda^{(2)}_{p,n}$ are defined in Eqs.EQ.(\ref{sum1}) and Eq.(\ref{sum2}), respectively.

Based on the previous preparation work, we are now in a position to determine $\mathbb{S}$ of $f$, where $\lambda^{(1)}_{p,n}$ and $\lambda^{(2)}_{p,n}$ are defined in Eqs.\eqref{sum1} and \eqref{sum2}, respectively.

\begin{theorem}\label{Thm1}
Let $f$ be defined as above. If $\gcd(d,~p^n-1)=$ 2 (resp. 4), then the $(-1)$-differential spectrum of $f$ is
\begin{align*}
\mathbb{S}=\{{}_{-1}\omega_{0}&=\frac{1}{64}(37p^n+7\lambda^{(1)}_{p,n}-2\lambda^{(2)}_{p,n}+5),\\
 {}_{-1}\omega_{2}&=\frac{1}{32}(11p^n-7\lambda^{(1)}_{p,n}+2\lambda^{(2)}_{p,n}-21),\\
{}_{-1}\omega_{3}&=1,\\
 {}_{-1}\omega_{4}&=\frac{1}{64}(5p^{n}+7\lambda^{(1)}_{p,n}-2\lambda^{(2)}_{p,n}-27)\}
\end{align*}
(resp.
\begin{align*}
\mathbb{S}=\{{}_{-1}\omega_{0}&=\frac{1}{64}(37p^n+7\lambda^{(1)}_{p,n}-2\lambda^{(2)}_{p,n}-11),\\
{}_{-1}\omega_{2}&=\frac{1}{32}(11p^n-7\lambda^{(1)}_{p,n}+2\lambda^{(2)}_{p,n}-5),\\
{}_{-1}\omega_{3}&=1,\\
 {}_{-1}\omega_{4}&=\frac{1}{64}(5p^{n}+7\lambda^{(1)}_{p,n}-2\lambda^{(2)}_{p,n}-43)\}
\end{align*}
when $\eta(2)=\eta(3)=1$, $\eta(\frac{-1+\sqrt{-2}}{2})=-1$ or $\eta(2)=\eta(3)=-1$, and is
\begin{align*}
\mathbb{S}=\{{}_{-1}\omega_{0}&=\frac{1}{64}(37p^n+7\lambda^{(1)}_{p,n}-2\lambda^{(2)}_{p,n}-27),\\
 {}_{-1}\omega_{1}&=1,\\
 {}_{-1}\omega_{2}&=\frac{1}{32}(11p^n-7\lambda^{(1)}_{p,n}+2\lambda^{(2)}_{p,n}-21),\\
 {}_{-1}\omega_{4}&=\frac{1}{64}(5p^{n}+7\lambda^{(1)}_{p,n}-2\lambda^{(2)}_{p,n}+5)\}
\end{align*}
(resp.
\begin{align*}
\mathbb{S}=\{{}_{-1}\omega_{0}&=\frac{1}{64}(37p^n+7\lambda^{(1)}_{p,n}-2\lambda^{(2)}_{p,n}-43),\\
 {}_{-1}\omega_{1}&=1,\\
 {}_{-1}\omega_{2}&=\frac{1}{32}(11p^n-7\lambda^{(1)}_{p,n}+2\lambda^{(2)}_{p,n}-5),\\
 {}_{-1}\omega_{4}&=\frac{1}{64}(5p^{n}+7\lambda^{(1)}_{p,n}-2\lambda^{(2)}_{p,n}-11)\}
\end{align*}
otherwise.
\end{theorem}

\begin{corollary}
Let $f(x)=x^{\frac{p^n+3}{2}}$ over $\mathbb{F}_{p^n}$, where $p^n \equiv 1 \pmod 4$ and $p\neq3$ is an odd prime. Then ${}_{-1}\Delta_{f}=4$.
\end{corollary}

\begin{proof}
By Lemma \ref{bound}, we obtain
\[{}_{-1}\omega_{4}\geq\frac{1}{64}(5p^n+7\lambda^{(1)}_{p,n}-2\lambda^{(2)}_{p,n}-43)
\geq\frac{1}{64}(5p^n-18p^{\frac{n}{2}}-43)>0\]
when $p^n\geq 29$. The numerical results show that ${}_{-1}\omega_4\geq 1$ for $p^n\in\{13,17,25\}$.

The case $p^n=5$ can be calculated by computing directly the $(-1)$-differential spectrum of $f(x)=x^{4}$ over $\mathbb{F}_{5}$, which is
\[\mathbb{S}=\{{}_{-1}\omega_{0}=3,~{}_{-1}\omega_{2}=1,~{}_{-1}\omega_{3}=1\},\]
and hence $_{-1}\Delta_{f}=4$ due to $\gcd(d,~p^n-1)=4$.

Therefore, we conclude that the $_{-1}\Delta_{f}=4$. This completes the proof.
\end{proof}

\begin{example}
Let $p=7$ and $n=4$. Then $f(x)=x^{1202}$ over $\mathbb{F}_{7^{4}}$ is differentially $(-1,4)$-uniform with $(-1)$-differential spectrum
$$\mathbb{S}=\{_{-1}\omega_{0}=1374,~{}_{-1}\omega_{1}=1,~{}_{-1}\omega_{2}=852,~ {}_{-1}\omega_{4}=174\}.$$
\end{example}
\begin{example}
Let $p=29$ and $n=1$. Then $f(x)=x^{16}$ over $\mathbb{F}_{29}$ is differentially $(-1,4)$-uniform with $(-1)$-differential spectrum
$$\mathbb{S}=\{_{-1}\omega_{0}=16,~{}_{-1}\omega_{2}=11,~{}_{-1}\omega_{3}=1,~ {}_{-1}\omega_{4}=1\}.$$
\end{example}
The above two examples both are consistent with that computed directly by Magma program.

\section{The $c$-differential uniformity of $f(x)=x^{\frac{p^n+3}{2}}$ over $\mathbb{F}_{p^n}$}\label{D}
In this section, we give the following result on the $c$-differential uniformity of $f$, and we omit the proof here since it is similar to that of \cite[Theorem 12]{YMT}.

\begin{theorem}\label{c-differential uniform}
Let $f(x)=x^d$ be a power function over $\mathbb{F}_{p^{n}}$, where $p$ is an odd prime and $d=\frac{p^{n}+3}{2}$ is a positive integer. For $\pm 1\neq c \in \mathbb{F}_{p^{n}}$, we have $_c\Delta_{f}\leq 9$.
\end{theorem}

\section{Concluding remarks}\label{E}
In this paper, we first investigated the $(-1)$-differential spectrum of $f(x)=x^{\frac{p^{n}+3}{2}}$ over $\mathbb{F}_{p^n}$, where $p\neq3$ is an odd prime. In fact, we prove that if $p^n \equiv 3 \pmod 4$, then $_{-1}\Delta_{f}=3$ except for $p^n\in\{7,19,23\}$ where $f$ is an APcN function, and if $p^n \equiv 1 \pmod 4$, the $(-1)$-differential uniformity of $f$ is equal to 4. The results indicate that the $(-1)$-differential spectrum of $f$ has closely connection with two character sums $\lambda^{(1)}_{p,n}$ and $\lambda^{(2)}_{p,n}$. Meanwhile, the character sums can be evaluated by employing the theory of elliptic curves over finite fields. Finally, we obtained an upper bound of the $c$-differential uniformity of $f$.

\bibliographystyle{plain}
\bibliography{text}

\end{document}